\documentclass[11pt,letter]{article}

\usepackage[utf8]{inputenc}
\usepackage{amsmath}
\usepackage{amssymb}
\usepackage{amsfonts}
\usepackage{amscd}
\usepackage{latexsym}
\usepackage{wasysym}
\usepackage{graphicx}
\usepackage{xcolor}
\usepackage{hyperref}
\usepackage[all,cmtip]{xy}
\usepackage{bbold}

\usepackage{amssymb}
\usepackage{amsthm}
\usepackage{enumitem}
\usepackage{mathtools,bbm}
\usepackage{mathtools,bbm}
\usepackage{xpatch}
\usepackage{version,xspace}

\usepackage[margin=1in]{geometry}
\usepackage{hyperref}%
\hypersetup{colorlinks=true, linkcolor=blue, breaklinks=true, urlcolor=blue}
\usepackage{url,doi}

\usepackage[utf8]{inputenc}
\usepackage{amsfonts}
\usepackage{amsthm}
\usepackage{enumitem}
\usepackage{mathtools,bbm}
\usepackage{mathtools,bbm}
\usepackage{xpatch}

\usepackage{lipsum}
\usepackage{mathtools}
\usepackage{cuted}
\usepackage{mathtools,bbm}

\usepackage[caption=false,font=footnotesize]{subfig}
\usepackage[font=small]{caption}
\graphicspath{{./images/}{./fig}}

\usepackage{enumitem}
\setlist{nosep}

\newcommand{\R}{\mathbb{R}}

\renewcommand{\ker}{\mathrm{Ker}}

\newcommand{\aug}{\operatorname{aug}}

\newcommand{\until}[1]{\{1,\dots, #1\}}

\newtheorem{theorem}{Theorem}[section]
\newtheorem{corollary}[theorem]{Corollary}
\newtheorem{lemma}{Lemma}[section]
\newtheorem{proposition}[theorem]{Proposition}
\newtheorem{remark}[theorem]{Remark}

{
      \theoremstyle{plain}

}

\DeclareSymbolFont{bbold}{U}{bbold}{m}{n}
\DeclareSymbolFontAlphabet{\mathbbold}{bbold}
\newcommand{\vect}[1]{\mathbbold{#1}}

\newcommand{\diag}{\operatorname{diag}}

\newcommand{\rank}{\operatorname{rank}}

\newcommand{\setdef}[2]{\{#1 \; | \; #2\}}

\newcommand{\fPD}{\subscr{F}{PD}}
\newcommand{\fPDd}{\subscr{F}{PD-d}}
\newcommand{\fPDo}{\subscr{F}{PD-tv}}

\usepackage{algpseudocode,algorithm,algorithmicx}

\usepackage{hyperref}%
\hypersetup{colorlinks=true, linkcolor=blue, breaklinks=true, urlcolor=blue}

\newcommand{\subscr}[2]{#1_{\textup{#2}}}

\newcommand{\map}[3]{#1: #2 \rightarrow  #3}

\newcommand{\real}{\ensuremath{\mathbb{R}}}

\newcommand\norm[1]{\left\lVert#1\right\rVert}
\renewcommand{\norm}[1]{\|#1\|}

\newcommand\oprocendsymbol{\hbox{$\square$}}
\newcommand\oprocend{\relax\ifmmode\else\unskip\hfill\fi\oprocendsymbol}

\DeclareSymbolFont{bbold}{U}{bbold}{m}{n}
\DeclareSymbolFontAlphabet{\mathbbold}{bbold}

\newcounter{saveenum}

\usepackage[prependcaption,colorinlistoftodos]{todonotes}

\title{A contraction analysis of primal-dual dynamics in distributed and time-varying implementations}
\title{Contraction analysis of distributed and time-varying primal-dual dynamics}
\title{Distributed and time-varying primal-dual dynamics via contraction analysis}

\author{Pedro Cisneros-Velarde, Saber Jafarpour, Francesco Bullo
  \thanks{Pedro Cisneros-Velarde (pacisne@gmail.com), Saber Jafarpour and Francesco Bullo (\{saber,bullo\}@ucsb.edu) are with the 
    University of California,
    Santa Barbara.}
    }
\date{}

\begin{document}
\maketitle

\begin{abstract}
  In this note, we provide an overarching analysis of primal-dual dynamics associated to
  linear equality-constrained optimization problems 
   using contraction analysis.
  For the 
  well-known standard version of the problem: we establish
  convergence 
  under convexity 
  and 
  the contracting rate under strong convexity.
 Then, for a canonical 
 distributed optimization problem,
 we use partial contractivity to
 establish  
  global exponential
  convergence of its 
  primal-dual dynamics.
  As an application, 
  we propose a new distributed solver for the least-squares
  problem with the same convergence guarantees. 
  Finally, for 
  time-varying versions of both centralized and distributed
  primal-dual dynamics,
  we exploit their contractive
  nature to establish bounds on their tracking error.  To
  support our 
  analyses, we introduce novel results on
  contraction theory.
\end{abstract}

%


\section{Introduction}

\paragraph*{Problem statement and motivation}

Primal-dual (PD) dynamics are dynamical systems that solve 
constrained optimization problems. Their study can be traced back to many decades
ago~\cite{KJA-LH-HU:58} and has 
regained
interest since the last decade~\cite{DF-FP:10}. 
PD 
dynamics have been made popular due to their scalability and
simplicity.
 They have been widely adopted in engineering applications 
such as resource allocation problems
in power networks~\cite{JWSP-BKP-NM-FD:19}, frequency control in
micro-grids~\cite{EM-CZ-SL:14a}, solvers for linear
equations~\cite{PW-SM-JL-WR:19}, etc. 
In this note, we study optimization
problems with linear equality constraints. In general, PD 
dynamics
seek to find a saddle point of the associated Lagrangian function to the
constrained problem, which is characterized by the equilibria of the dynamics. 
%
%
For a general treatise of asymptotic stability of saddle points, 
we refer to
~\cite{AC-BG-JC:17} and references therein. However, despite the long history of study and application, there are very recent 
studies on PD 
dynamics related to linear equality constraints 
further 
studying 
different dynamic properties such as: exponential  convergence under different convexity assumptions
~\cite{GQ-NL:19,XC-NL:19} and contractivity properties~\cite{HDN-TLV-KT-JJS:18}. 
%

We are particularly interested in studying primal-dual dynamics in distributed and time-varying optimization problems. We refer to the recent survey~\cite{TY-XY-JW-DW:19} 
for an  overview of the long-standing interest on distributed 
optimization. 
Of particular interest is to provide strong convergence guarantees such as global (and exponential)  convergence for the distributed solvers. We aim to provide them 
using contraction theory.
%
%
Time-varying optimization has 
found applications 
in system identification, signal detection, robotics, 
traffic management, etc.~\cite{MF-SP-VMP-AR:18,CS-MY-GH:17}. The goal is to employ a dynamical system 
able to track the time-varying optimal solution up to some bounded error in real time. Although different dynamics have been proposed to 
both time-varying centralized~\cite{MF-SP-VMP-AR:18} and distributed 
problems~\cite{CS-MY-GH:17,SR-WR:17}, to the best of our knowledge, there has not been a characterization of the PD 
dynamics in such application contexts. The importance of PD 
algorithms is their simplicity of implementation, i.e., they do not require more complex information structures like the inverse of the Hessian of the system at all times, as in~\cite{MF-SP-VMP-AR:18} and~\cite{SR-WR:17} for the centralized and distributed cases respectively. 
However, 
simplicity may come with 
a possible trade-off in the tracking error.

Contraction is valuable in practice because it introduces strong stability and robustness guarantees. 
For example, it implies input-to-state stability for systems subject to state-independent disturbances. It also 
guarantees fast correction after transient perturbations to the trajectory of the solution, 
since initial conditions are forgotten. Moreover, a contractive system may be robust towards structural perturbations on the vector field, 
e.g., when a non-convex term is added to the objective function. Finally, 
contraction 
guarantees stable numerical discretizations with geometric convergence rates, an ideal situation for practical implementations. 
%
All these properties are transparent to whether the system is time-varying or not. 
All of this 
motivates a contraction analysis of PD 
algorithms 
%
in contrast to the prevalent Lyapunov or invariance analysis in the literature. 

\paragraph*{Literature review}

The recent works~\cite{GQ-NL:19,HDN-TLV-KT-JJS:18,XC-NL:19} study convergence properties of 
PD 
dynamics under different assumptions on the objective function. 
In 
distributed optimization, solvers based on PD 
dynamics are fairly recent, e.g., ~\cite{JW-NE:11,JC-SKN:19,TY-XY-JW-DW:19}. 
An application of distributed optimization of current interest - as seen in the recent survey~\cite{PW-SM-JL-WR:19} - is the \emph{distributed least-squares problem} for solving 
an over-determined system of linear equations. To the best of our knowledge, solvers for this problem (in continuous-time) with exponential global convergence 
are still missing in the literature.
%

Finally, this paper is related to contraction theory, a mathematical tool to analyze 
incremental stability
~\cite{WL-JJES:98,MV:02}. An introduction and survey 
can be found in~\cite{ZA-EDS:14b}. A variant of contraction theory, \emph{partial contraction}~\cite{QCP-JJS:07,MdB-DF-GR-FS:16}, analyzes the convergence to linear subspaces 
and 
has been used in the synchronization analysis of diffusively-coupled network systems~\cite{QCP-JJS:07,ZA-EDS:14}; however, its application to distributed algorithms is still missing, and our paper provides such contribution. 

\paragraph*{Contributions}
In this paper we consider the PD dynamics associated to optimization problems with a 
twice differentiable and strongly convex
objective function 
and linear equality constraints. We use 
contraction theory to perform an overarching study of PD dynamics in a variety of implementations and applications; see Fig.~\ref{fig:rel}. In particular: 

(i) We introduce new theoretical results 
of how 
\emph{non-expansiveness} and \emph{partial contraction} can imply exponential convergence 
to a 
point 
in a subspace of equilibria. 

(ii) For the standard and distributed PD dynamics, we prove: 1) convergence under 
non-expansiveness when the objective function is convex; 2) contraction for the standard problem and partial contraction for the distributed one 
in the 
strongly convex
case, with closed-form exponential global convergence rates. The analysis in result 1) is novel, since it uses the new results introduced in (i). 
Compared to the work~\cite{HDN-TLV-KT-JJS:18} that also shows contraction for the standard PD, our proof method provides an explicit closed-form expression of the system's contraction rate. Our exponential convergence rate is different from the one by~\cite{GQ-NL:19} via Lyapunov analysis, and 
both rates 
cannot be compared 
without extra 
assumptions on the numerical relationships among various parameters associated to the 
objective function or constraints. 
Moreover, we propose using the \emph{augmented Lagrangian} 
in order to achieve contraction when the objective function is 
only 
convex. 
%
In the case of distributed optimization, 
there exist other 
solvers that show exponential convergence, e.g., as in~\cite{SK-JC-SM:15,SL-LYW-GY:19}, but none of these have 
contractivity. 
(iii) We propose a new solver for the distributed least-squares problem based on 
PD dynamics, and use our results in (ii) to prove its convergence. 
Compared to the recent work~\cite{YL-CL-BDOA-GS:19}, our new model exhibits global convergence; and compared to the recent work
~\cite{YL-YL-DOA-GS:18}, ours exhibits exponential convergence 
and has a simpler structure. 
%

(iv) 
We 
characterize the performance of PD dynamics associated to time-varying versions of both standard and distributed optimization problems in terms of the problems' parameters. In particular, we prove the tracking error to the time-varying solutions is uniformly ultimately bounded (UUB) in either case and that the 
bound decreases as the contraction rate increases --- these results, to the best of our knowledge, are novel. Our analysis builds upon the contraction results in contribution (ii). 

\paragraph*{Paper organization} 
Section~\ref{sec:prelim} has notation and preliminary concepts. Section~\ref{sec:contr-th} has results on contraction theory. Section~\ref{sec:basic-prob} analyzes contractive properties of the standard PD dynamics. 
The contractive analysis of distributed (with the least-squares problem application) and time-varying versions of PD dynamics are in Sections~V and~VI respectively.
%
Section~\ref{sec:concl} is the conclusion.

\section{Preliminaries and notation}
\label{sec:prelim}

\subsection{Notation, definitions and useful results}
Consider $A\in\R^{n\times{n}}$, 
then $\sigma_{\min}(A)$ denote its minimum
singular value and $\sigma_{\max}(A)$ its maximum one. 
%
%
If $A$ has only
real eigenvalues, let $\lambda_{\max}(A)$ be its 
maximum eigenvalue. $A$
  is an orthogonal projection if it is symmetric
  and $A^2 = A$. Let $\norm{\cdot}$ denote any norm, and $\norm{\cdot}_p$ denote the $\ell_p$-norm. When the argument of a norm is a matrix, we refer to its respective induced norm. 
The matrix measure associated 
to 
$\norm{\cdot}$ is $\mu(A) = \lim_{h\to 0^{+}}\frac{\|I+hA\|-1}{h}$; e.g., 
the one 
associated to the $\ell_2$-norm is 
$\mu_2(A)=\lambda_{\max}((A+A^\top)/2)$~\cite{ZA-EDS:14b}. 
Given invertible $Q\in\R^{n\times n}$, let $\norm{\cdot}_{2,Q}$ be the weighted $\ell_2$-norm $\norm{x}_{2,Q}=\norm{Qx}_2$
, $x\in\R^n$, and whose associated matrix measure is 
$\mu_{2,Q}(A)=\mu_{2}(QAQ^{-1})$~\cite{ZA-EDS:14b}. 
  
Let $I_n$ be the $n\times n$ identity matrix, $\vect{1}_n$ and
$\vect{0}_n$ be the all-ones and all-zeros column vector with $n$
entries respectively. Let $\diag(X_1,\dots,X_N)\in\R^{\sum^N_{i=1}n_i\times\sum^N_{i=1}n_i}$ be the 
block-diagonal matrix with elements 
$X_i\in\R^{n_i\times n_i}$. 
Let $\R_{\geq 0}$ be the set of non-negative real numbers. Given $x_i\in\R^{k_i}$, 
let $(x_1,\dots,x_N)=\begin{bmatrix}x_1^\top &\dots&x_N^\top\end{bmatrix}$.

Consider a differentiable function $f:\R^n\to\R^n$. We say $f$ is \emph{Lipschitz smooth} with constant $K_1>0$ if $\norm{\nabla f(x)-\nabla f(y)}_2\leq K_1\norm{x-y}_2$ for any $x,y\in\R^n$; and \emph{strongly convex} with constant $K_2>0$ if $K_2\norm{x-y}_2^2\leq (\nabla f(x)-\nabla f(y))^\top(x-y)$ for any $x,y\in\R^n$. Assuming $f$ is twice differentiable, these two conditions are equivalent to $\nabla^2 f(x)\preceq K_1I_n$ and $K_2I_n\preceq\nabla^2 f(x)$ for any $x\in R^n$, respectively.
The proof of the next proposition is found in the Appendix.

\begin{proposition}
  \label{prop-neg-eig} For a full-row rank matrix $A\in\R^{m\times{n}}$,
  $B=B^\top\in R^{n\times n}$, and $b_2\geq{b_1}>0$ such that $b_2
  I_n\succeq B \succeq b_1 I_n\succ0$, the matrix $\begin{bmatrix} -B
    &-A^\top\\ A & \vect{0}_{m\times m}
\end{bmatrix}$ is Hurwitz.
\end{proposition}

\subsection{Review of basic concepts on contraction theory}

Consider the dynamical system 
 $\dot{x}=f(x,t)$
with $x\in\R^n$. 
Let $t\mapsto\phi(t,t_0,x_0)$ be the trajectory of
  the system 
  starting from $x_0\in\R^n$ at time $t_0\geq 0$.
%
Consider the system satisfies $\norm{\phi(t,t_0,x_0)-\phi(t,t_0,y_0)}\leq \norm{x_0-y_0}e^{-c(t-t_0)}$, for any $x_0,y_0\in\R^n$ and any $t_0\in\real_{\ge 0}$. 
We say it is \emph{contractive} with
respect to $\norm{\cdot}$ when $c>0$, and \emph{non-expansive} 
when $c=0$.
%
%
A time-invariant contractive system has a unique
equilibrium point.
%
%
%
Now, assume the Jacobian of the system, 
i.e., $Df(x,t)$, satisfies:
$\mu(Df(x,t))\leq -c$ for any $(x,t)\in
\R^n\times\real_{\ge 0}$, with $\mu$ being the matrix measure associated to $\norm{\cdot}$ and constant $c\geq 0$. 
Then, 
this system 
has contraction rate $c$ with respect to
$\norm{\cdot}$. 
Now, assume the
system 
has a flow-invariant linear subspace
$\mathcal{M}=\setdef{x\in\R^n}{Vx=\vect{0}_k}$ with
$V\in\R^{k\times{n}}$ being 
full-row rank 
with orthonormal
rows. Then the system is \emph{partially contractive} with respect to
$\norm{\cdot}$ and 
$\mathcal{M}$ if there exists
$c>0$ such that, for any $x_0\in\R^n$ and $t_0\in \real_{\geq 0}$, the 
system 
satisfies
$\norm{V\phi(t,t_0,x_0)}\leq \norm{Vx_0}e^{-c(t-t_0)}$. When
$c=0$, 
the system is \emph{partially non-expansive} with
respect to $\mathcal{M}$~\cite{QCP-JJS:07}. Consequently, a partially
contractive system has any of its trajectories approaching $\mathcal{M}$
with exponential rate; 
%
and a partially non-expansive one has any of its trajectories at a
non-increasing distance from $\mathcal{M}$. 

Pick a symmetric positive-definite $P\in\R^{n\times{n}}$ and a scalar $c>0$, then 
$\mu_{2,P^{1/2}}(Df(x,t))\leq -c$ for all $(x,t)\in\R^n\times{\R_{\geq 0}}$ is equivalent to $f$ satisfying the \emph{integral contractivity condition}, i.e., for every $x,y\in\R^n$ and $t\geq 0$,
$(y-x)^\top P(f(x,t)-f(y,t))\leq -c\norm{x-y}_{2,P^{1/2}}^2$.
%
%
    
\section{Theoretical contraction results}
\label{sec:contr-th}
The next 
result 
will be used throughout the paper.

\begin{theorem}[Results on partial contraction]
\label{th-contr-sub}
Consider the system $\dot{x}=f(x,t)$, $x\in\R^n$, 
with a flow-invariant $\mathcal{M}=\setdef{x\in\R^n}{Vx=\vect{0}_k}$ with $V\in\R^{k\times{n}}$ being a full-row rank matrix with orthonormal rows. Assume $\mu(VDf(x,t)V^\top)\leq -c$ for any $(x,t)\in\R^n\times\R_{\geq 0}$, some constant $c\geq 0$ and some matrix measure $\mu$.
\begin{enumerate}
\item\label{thss_1}If $c>0$, then the system is partially contractive with
  respect to $\mathcal{M}$ and every trajectory exponentially converges to
  the subspace $\mathcal{M}$ with rate $c$.
\item\label{thss_1a}If $c=0$ and $\mu(VDf((I_n-V^\top V)x,t)V^\top)<0$ for
  any $(x,t)\in\R^n\times \real_{\ge 0}$, then the system is partially
  non-expansive with respect to $\mathcal{M}$ and every trajectory
  converges to the subspace $\mathcal{M}$.
  \setcounter{saveenum}{\value{enumi}}
\end{enumerate}
Moreover, assume that one of the conditions in parts~\ref{thss_1} and~\ref{thss_1a} holds and $\mathcal{M}$ is a set of equilibrium points. If the system is non-expansive, then 
\begin{enumerate}\setcounter{enumi}{\value{saveenum}}
\item\label{thss_2}every trajectory of the system converges to an equilibrium point, and if $c>0$, then it does it with exponential rate $c$. 
\setcounter{saveenum}{\value{enumi}}
\end{enumerate}
\end{theorem}
 
\begin{remark}
  Statement~\ref{thss_1} in Theorem~\ref{th-contr-sub} was proved
  in~\cite{QCP-JJS:07}.  To the best of our knowledge,
  statements~\ref{thss_1a} and~\ref{thss_2} are novel.
\end{remark}

\begin{proof}[Proof of Theorem~\ref{th-contr-sub}]
It is easy to check that $V^\top V$ is an orthogonal projection matrix onto $\mathcal{M}^\perp$; and that $U:=I_n-V^\top V$ is also an orthogonal projection matrix onto $\mathcal{M}$. 
%
Using these results, we can express the given system as $\dot{x}=f(Ux+V^\top Vx,t)$. Now, we set $z:=Vx$, and observe that $x(t)$ converges to $\mathcal{M}$ if and only if $z(t)$ converges $\vect{0}_{k}$. Then, using this change of coordinates, we obtain the system:
\begin{equation}
\label{syst_z}
\dot{z}=Vf(Ux+V^\top z,t).
\end{equation}
It has been proved in~\cite[Theorem~3]{QCP-JJS:07} that 
$z^*=\vect{0}_{k}$ is an equilibrium point for the
system~\eqref{syst_z}. 

To prove~\ref{thss_1a}, assume that
$\mu(VDf(x,t)V^\top)=0$ for any $(x,t)\in\R^n\times \real_{\ge 0}$; i.e.,
that the system~\eqref{syst_z} is non-expansive. Now, if we
assume that $\mu(VDf(Ux,t)V^\top)<0$ for any $(x,t)\in\R^n\times
\real_{\ge 0}$, then by the Coppel's inequality~\cite{WAC:1965}, the
fixed point $z^*=\vect{0}_k$ is locally exponentially stable. Now, we can use 
a generalization of~\cite[Lemma~6]{EL-GC-KS:14}, namely Lemma~\ref{prop-weak-conv} (proof found in the Appendix),  
to establish the convergence of $z(t)$ to $z^*$. This finishes the proof for~\ref{thss_1a}.

Now, we prove statement~\ref{thss_2}. 
  Let $t\mapsto x(t)$ be a trajectory of the dynamical system. For
  every $t\in \real_{\ge 0}$, $\big(I_n - V^{\top}V\big)x(t)$ is the orthogonal projection
  of $x(t)$ onto the subspace $\mathcal{M}$ and it is an equilibrium
  point. Since the dynamical system is non-expansive, we have
$\|x(s) - (I_n - V^{\top}V)x(t)\| \le \|x(t) -
    (I_n - V^{\top}V)x(t)\|=\|V^{\top}Vx(t)\|$, for all $s\ge t$. 
  This implies that, for every $t\in \real_{\ge 0}$ and every $s\ge t$, the point $x(s)$ remains
  inside the closed ball $\overline{B}(x(t),
  \|V^{\top}Vx(t)\|)$. Therefore, for every $t\ge 0$, the point
  $x(t)$ is inside the set $C_t$ defined by 
  $C_t = \mathrm{cl}\big(\bigcap_{\tau\in [0,t]} \overline{B}(x(\tau),
  \|V^{\top}Vx(\tau)\|)\big)$. 
It is easy to see that, for $s\ge t$, we have $C_{s}\subseteq
C_t$. This implies that the family $\{C_t\}_{t\in [0,\infty)}$ is a
nested family of closed subsets of $\real^n$. 
Moreover, by parts~\ref{thss_1} and~\ref{thss_1a}, we have that $\lim_{t\to\infty}\|V^{\top}Vx(t)\|\to 0$ as $t\to\infty$, which in turn results in $\lim_{t\to\infty}\mathrm{diam}(C_t) = 0$, with convergence rate $c$ for the case $c>0$ because of $\mathrm{diam}(C_t)=\|V^{\top}Vx(t)\|\le
2e^{-ct}\|V^{\top}Vx(0)\|$.
%
Thus, by the Cantor
Intersection Theorem~\cite[Lemma 48.3]{JM:00}, there exists $x^*\in \real^n$ such that
$\bigcap_{t\in [0,\infty)} C_t = \{x^*\}$. We first show that
$\lim_{t\to \infty} x(t) = x^*$. Note that $x^*,x(t)\in C_t$, for every
$t\in \real_{\ge 0}$. This implies that 
$\|x(t) - x^*\|\le
\mathrm{diam}(C_t)$. 
This in turn
means that $\lim_{t\to\infty} \|x(t) - x^*\|=0$ and $t\mapsto x(t)$
converges to $x^*$, with convergence rate $c$ for the case $c>0$. On the other hand, by part~\ref{thss_1}, the
trajectory $t\mapsto x(t)$ converges to the subspace $\mathcal{M}$. Therefore, $x^*\in \mathcal{M}$ and $x^*$ is an
equilibrium point of the dynamical system. 
This completes the proof for statement~\ref{thss_2}.
%
\end{proof}

%
   
\section{The standard 
optimization problem}
\label{sec:basic-prob}
We consider the constrained optimization problem:
\begin{equation}\label{eq:opt}
	\min_{x\in \real^n} \; f(x)\quad\text{ subject to }\quad Ax= b
%
\end{equation}
with the following standing assumptions: 
$A\in\R^{k\times n}$, $k<n$, $b\in\real^k$, $A$ is full-row rank,
and $\map{f}{\real^n}{\real}$ is convex and  twice differentiable.
%

Associated to the optimization problem~\eqref{eq:opt} is the
\emph{Lagrangian function} 
$\mathcal{L}(x,\nu)=f(x)+\nu^\top(Ax-b)$ and the \emph{primal-dual dynamics}
\begin{equation}\label{eq:pd}
  \begin{bmatrix}
    \dot{x}\\\dot{\nu}
  \end{bmatrix}=
  \begin{bmatrix}
    -\frac{\partial \mathcal{L}(x,\nu)}{\partial x}\\
    \frac{\partial \mathcal{L}(x,\nu)}{\partial \nu}
  \end{bmatrix}=
  \begin{bmatrix}
    -\nabla f(x)-A^\top\nu\\
    Ax-b
  \end{bmatrix}.
\end{equation}


We introduce two possible sets of
assumptions:
\begin{enumerate}[label=\textup{(A\arabic*)}]
\item\label{ass:1} the primal-dual dynamics~\eqref{eq:pd} have an
  equilibrium $(x^*,\nu^*)$ and $\nabla^2f(x^*)\succ \vect{0}_{n\times{n}}$;
\item\label{ass:2} the function $f$ is strongly convex with constant
  $\ell_{\inf}>0$ and Lipschitz smooth with constant $\ell_{\sup}>0$, and,
  for $0<\epsilon<1$, we define
\end{enumerate}
\begin{equation}
\begin{split}  
    &\alpha_{\epsilon} :=  \frac{\epsilon \ell_{\inf}}
          {\sigma^2_{\max}(A)+\frac{3}{4}\sigma_{\max}(A)\sigma^2_{\min}(A)+\ell_{\sup}^2}>0
          \\
          %
		&P := \begin{bmatrix} I_n&\alpha_\epsilon\;
            A^\top\\ \alpha_\epsilon\; A &I_k
         \end{bmatrix} \in\real^{(n+k)\times(n+k)}.
         \label{def:P}
\end{split}
\end{equation}
\begin{theorem}[Contraction analysis of primal-dual dynamics]\label{th-1}
  Consider the constrained optimization problem~\eqref{eq:opt}, its
  standing assumptions, and its associated primal-dual
  dynamics~\eqref{eq:pd}.
  \begin{enumerate}
  \item\label{1-1} The primal-dual dynamics
    is non-expansive with respect to $\norm{\cdot}_2$ and, if
    Assumption~\ref{ass:1} holds, then $(x^*,\nu^*)$ is globally asymptotically stable.
  \item \label{1-2} Under Assumption~\ref{ass:2},
  \begin{enumerate}
  \item\label{1-2-1} the primal-dual dynamics are contractive with respect
    to $\norm{\cdot}_{2,P^{1/2}}$ with contraction rate
    \begin{equation}
      \label{contr_1}
      \alpha_\epsilon\frac{3}{4}\frac{\sigma_{\max}(A)\sigma^2_{\min}(A)}{\sigma_{\max}(A)+1},
      \quad\text{and }
    \end{equation}
  \item\label{1-2-2} there exists a unique globally exponentially stable
    equilibrium point $(x^*,\nu^*)$, and $x^*$
    is the unique solution to the optimization
    problem~\eqref{eq:opt}.
  \end{enumerate}
  \end{enumerate}
\end{theorem}

\begin{proof}
Let $(\dot{x},\dot{\nu})^\top:=\fPD(x,\nu)$. 
Then, $D\fPD(x,\nu)=\begin{bmatrix}
  -\nabla^2f(x) & -A^\top\\ A & 0
\end{bmatrix}$, and so 
$\mu_2(D\fPD(x,\nu))=\lambda_{\max}\left((D\fPD(x,v)+D\fPD(x,v)^\top)/2\right)
=\lambda_{\max}\left(\diag(
-\nabla^2f(x),\vect{0}_{k\times k})\right)=0$ 
for any $(x,\nu)\in\R^n\times\R^m$, because of convexity $\nabla^2f(x)\succeq 0$, which implies 
the system is non-expansive.
For the second part of statement~\ref{1-1}: Proposition~\ref{prop-neg-eig} implies $D\fPD(x^*,\nu^*)$ is Hurwitz since $\nabla^2f(x^*)\succ 0$, and the proof follows from 
a simple generalization of~\cite[Lemma~6]{EL-GC-KS:14} (its proof can be found in the Appendix).

Now, we prove statement~\ref{1-2}. 
Define $P=\begin{bmatrix}
I_n&\alpha A^\top\\
\alpha A &I_k
\end{bmatrix}$ which is a positive-definite matrix when 
\begin{equation}
\label{dd0}
0<\alpha<\frac{1}{\sigma_{\max}(A)}. 
\end{equation}
We plan to use the integral contractivity condition to show that system~\eqref{eq:pd} is contractive with respect to norm 
$\norm{\cdot}_{2,P^{1/2}}$. 
Thus, we need to show 
\begin{align*}
\eta&:=\begin{bmatrix}
x_1-x_2\\\nu_1-\nu_2
\end{bmatrix}^\top P(\fPD(x_1,\nu_1)-\fPD(x_2,\nu_2))\\
&+c\begin{bmatrix}
x_1-x_2\\\nu_1-\nu_2
\end{bmatrix}^\top P\begin{bmatrix}
x_1-x_2\\\nu_1-\nu_2
\end{bmatrix}\leq 0
\end{align*}
for any $x_1,x_2\in\R^n$ and $\nu_1,\nu_2\in\R^m$, and some constant $c>0$ which will be the contraction rate. After completing squares, using the strong convexity and Lipschitz smoothness of $f$, along with $\sigma^2_{\min}(A)I_k\preceq AA^\top$ and $A^\top A\preceq \sigma^2_{\max}(A)I_n$, we obtain
\begin{align*}
%
%
\eta&\leq -\left((3\alpha/4)\sigma^2_{\min}(A)-c-c\alpha\right)\norm{\nu_1-\nu_2}_2^2-(\ell_{\inf}\\
&\;-\alpha\sigma^2_{\max}(A)-c-\alpha\ell^2_{\sup}-c\alpha\sigma^2_{\max}(A))\norm{x_1-x_2}_2^2\\
&\;-\alpha c\norm{(\nu_1-\nu_2)-A(x_1-x_2)}_2^2.
\end{align*}
Set $c=D\alpha$ for some $D>0$. Then, to ensure that $\eta\leq 0$, we need to ensure 
\begin{align}
&\frac{3\alpha}{4}\sigma^2_{\min}(A)-D\alpha-D\alpha^2\ge 0,\label{dd}\\
&\ell_{\inf}-\alpha\sigma^2_{\max}(A)-D\alpha-\alpha\ell^2_{\sup}-D\alpha^2\sigma^2_{\max}(A)\ge
  0.\label{dd1}
\end{align}
Now, to ensure inequality~\eqref{dd} holds, using the inequalities~\eqref{dd0}, 
it is easy to see that it suffices to ensure that 
\begin{equation}
\label{dd2}
\frac{3\sigma_{\max}(A)\sigma^2_{\min}(A)}{4(\sigma_{\max}(A)+1)}>D.
\end{equation}
Now, using inequalities~\eqref{dd0} and~\eqref{dd2}, we obtain: 
$\ell_{\inf}-\alpha\sigma^2_{\max}(A)-D\alpha-\alpha\ell^2_{\sup}-D\alpha^2\sigma^2_{\max}(A)>\ell_{\inf}-\alpha(\sigma^2_{\max}(A)+\frac{3}{4}\sigma_{\max}(A)\sigma_{\min}^2(A)+\ell_{\sup}^2)$ 
and so, to ensure inequality~\eqref{dd1} holds, it suffices that 
\begin{equation}
\label{dd3}
\frac{\ell_{\inf}}{\sigma^2_{\max}(A)+\frac{3}{4}\sigma_{\max}(A)\sigma^2_{\min}(A)+\ell_{\sup}^2}>\alpha.\end{equation}
Now, the parameter $\alpha$ needs to satisfy inequalities~\eqref{dd0} and~\eqref{dd3}; however, \eqref{dd3} implies~\eqref{dd0} because the inequality $\pi_1^2+\pi_2^2\geq 2\pi_1\pi_2$ for $\pi_1,\pi_2>0$ let us conclude that $\frac{\ell_{\sup}}{\sigma^2_{\max}(A)+\ell^2_{\sup}}\leq\frac{1}{2\sigma_{\max}(A)}$. 
Finally, $c$ must be less than the multiplication of the left-hand sides of the inequalities~\eqref{dd2} and~\eqref{dd3}, which proves statement~\ref{1-2-1}. 

Now, since the dynamics are contractive, there must exist a unique globally exponentially stable equilibrium point which also satisfies the (sufficient and necessary) KKT conditions of optimality for the optimization problem~\eqref{eq:opt}, thus proving statement~\ref{1-2-2}.
\end{proof}


\begin{figure}[t]
  \centering
\includegraphics[width=0.85\linewidth]{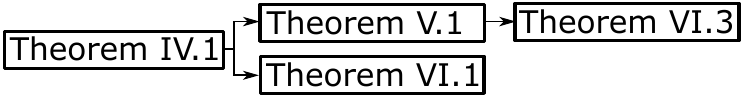}
\caption{Dependency relationship among the main theorems in this note. An arrow from $A$ to $B$ means $A$ is used to prove $B$.
}\label{fig:rel}
\end{figure} 
\begin{remark}\label{remark-1}
Theorem~\ref{th-1} is a fundamental
  building block for the rest of results in this paper as seen in Fig.~\ref{fig:rel} and therefore, it was necessary to provide a comprehensive proof
  using the integral contractivity condition that could provide an explicit estimate of the contraction rate (as opposed to the different proof in~\cite{HDN-TLV-KT-JJS:18}).
  
    Throughout this note, 
  the Lipschitz smoothness and strong convexity of $f$ 
  are used to prove contraction.  
  However, the latter is relaxed in Corollary~\ref{th-3}.
\end{remark}

For the case of convex $f$, Theorem~\ref{th-1} does not state convergence - nor contraction - without additional assumptions; indeed, 
oscillations may appear 
and convergence to
the saddle points is not guaranteed~\cite{DF-FP:10}.
%
%
%
%
In order to still be able to use Theorem~\ref{th-1} in this case, we 
consider
a modification to the
Lagrangian, 
known as the \emph{augmented
  Lagrangian}~\cite{RTR:76}: 
$\mathcal{L}_{\aug}(x,\nu)=\mathcal{L}(x,\nu)+\frac{\rho}{2}\norm{Ax-b}^2_2$ 
with gain $\rho> 0$. Its associated \emph{augmented primal-dual dynamics}
become
\begin{equation}\label{eq:pd_aug}
\begin{bmatrix}
\dot{x}\\\dot{\nu}
\end{bmatrix}=
\begin{bmatrix}
-\nabla f(x)-A^\top\nu-\rho A^\top A x+\rho A^\top b\\
Ax-b
\end{bmatrix}
\end{equation} 
and 
have the same equilibria as the original one
in~\eqref{eq:pd}. 
We introduce two possible sets of assumptions:
\begin{enumerate}[label=\textup{(A\arabic*)}]\setcounter{enumi}{2}
\item\label{ass:3} the primal-dual dynamics~\eqref{eq:pd} have an
  equilibrium $(x^*,\nu^*)$, $\nabla^2f(x^*)\succeq
\vect{0}_{n\times{n}}$, and $\ker(\nabla^2f(x^*))\cap\ker(A)=\{\vect{0}_n\}$;
  
\item\label{ass:4} $\ker(\nabla^2f(x))\cap\ker(A)=\{\vect{0}_n\}$ for any
  $x\in\R^n$ and $f$ is Lipschitz smooth with constant $\ell_{\sup}>0$,
  and, for $0<\epsilon<1$, we define
\end{enumerate}
  \begin{equation}
  \begin{split}
    &\bar{\alpha}_{\epsilon} :=  \frac{\epsilon \rho \sigma_{\min}^2(A)}
          {(1+\rho)\sigma^2_{\max}(A)+\frac{3}{4}\sigma_{\max}(A)\sigma^2_{\min}(A)+\ell_{\sup}^2}\\
     &\bar{P} := \begin{bmatrix} I_n&\bar{\alpha}_\epsilon\;
            A^\top\\ \bar{\alpha}_\epsilon\; A &I_k
         \end{bmatrix} \in\real^{(n+k)\times(n+k)}.
         \label{def:Pbar}
    \end{split}
  \end{equation}
\begin{corollary}[Contraction analysis of the augmented primal-dual dynamics]
  \label{th-3}
  Consider the constrained optimization problem~\eqref{eq:opt}, its
  standing assumptions, and its associated augmented primal-dual
  dynamics~\eqref{eq:pd_aug} with $\rho>0$.
  \begin{enumerate}
  \item\label{3-1} Under Assumption~\ref{ass:3}, the augmented primal-dual
    dynamics are non-expansive with respect to $\norm{\cdot}_2$ and
    $(x^*,\nu^*)$ is globally asymptotically stable.
  \item \label{3-2} Under Assumption~\ref{ass:4},
    \begin{enumerate}
    \item\label{3-2-1} the augmented primal-dual dynamics are contractive
      with respect to $\norm{\cdot}_{2,\bar{P}^{1/2}}$ with contraction rate
      \begin{equation}
      \label{contr_1bar}
      \bar{\alpha}_\epsilon\frac{3}{4}\frac{\sigma_{\max}(A)\sigma^2_{\min}(A)}{\sigma_{\max}(A)+1},
      \quad\text{and }
    \end{equation}
  \item\label{3-2-2} there exists a unique globally exponentially stable
    equilibrium point $(x^*,\nu^*)$ for the augmented primal-dual dynamics
    and $x^*$ is the unique solution to the constrained optimization
    problem~\eqref{eq:opt}.
  \end{enumerate}
  \end{enumerate}
\end{corollary}
%
%
%
\begin{proof}
The proof follows directly from Theorem~\ref{th-1}. For statement~\ref{3-1}, note that $\ker(\nabla^2f(x^*))\cap\ker(A)=\{\vect{0}_n\}$ implies that $\nabla^2f(x^*)+\rho A^\top A\succ \vect{0}_{n\times{n}}$ for the Jacobian of the system $$\begin{bmatrix}
-\nabla^2f(x)-\rho A^\top A & -A^\top\\
A & 0
\end{bmatrix}.$$
For statement~\ref{3-2}, note that $\ker(\nabla^2f(x))\cap\ker(A)=\{\vect{0}_n\}$ for any $x\in\R^n$ implies that 
$x\mapsto f(x)+\frac{\rho}{2}x^\top A^\top Ax$ is Lipschitz smooth with constant $\ell_{\sup}^2+\rho\sigma^2_{\max}(A)>0$ and strongly convex with constant $\rho\sigma_{\min}^2(A)>0$.
\end{proof}

\begin{remark}[Augmented Lagrangian and contraction] The benefit of
    using the augmented Lagrangian is that, unlike the conditions in
    Theorem~\ref{th-1}, the resulting primal-dual dynamics may be contractive despite
    $f$ being only convex.
\end{remark}

\section{Distributed algorithms}
\label{sec:ditrib-prob}
\label{dec-dist}

%
We study a popular distributed implementation for solving an unconstrained optimization problem
~\cite{TY-XY-JW-DW:19}. 
%
We want to solve the problem $\min_{x\in \R^{n}}f(x)=\sum_{i=1}^Nf_i(x)$ with $f_i:\R^n\to\R$ convex. 
Let $\mathcal{G}$ be an undirected connected 
interaction graph between $N$ distinct agents. 
Let $\mathcal{N}_i$ be the neighborhood of node $i$ and $L$ be the 
Laplacian matrix of $\mathcal{G}$. Let $x^i\in\R^n$ be the state 
associated to 
agent $i$, and let $\textbf{x}=(x^1,\dots,x^N)^\top$. Then, 
the problem becomes: 
\begin{equation}\label{eq:opt_dec_cent}
%
%
\begin{split}
	\min_{\textbf{x}\in \R^{nN}} \quad &\sum_{i=1}^Nf_i(x^i)\\
	&(L\otimes I_n)\textbf{x}= \vect{0}_{nN}\\
\end{split}.
\end{equation}
The associated \emph{distributed primal-dual dynamics} are 
\begin{equation}\label{eq:pd_dec}
\begin{split}
	\dot{x}^i&
	=-\nabla_{x^i}f_i(x^i)-\sum_{j\in\mathcal{N}_i}(\nu^j-\nu^i)\\
	\dot{\nu}^i&=\sum_{j\in\mathcal{N}_i}(x^j-x^i)
\end{split}
\end{equation}
for $i\in\until{N}$. In system~\eqref{eq:pd_dec}, 
any agent only uses information from herself and the set of her neighbors.

To study this system, we introduce two possible sets of
assumptions:
\begin{enumerate}[label=\textup{(A\arabic*)}]\setcounter{enumi}{4}
\item\label{ass:5} $\min_{x\in \R^{n}}f(x)$ has a solution $x^*$ and $\nabla^2f_i(x^*)\succ \vect{0}_{n\times{n}}$ for any $i\in\{1,\dots,N\}$;
\item\label{ass:6} $\min_{x\in \R^{n}}f(x)$ has a solution $x^*$ and the function $f_i$ is strongly convex with constant $\ell_{\inf,i}>0$ and Lipschitz smooth with constant $\ell_{\sup,i}>0$ for any $i\in\{1,\dots,N\}$, with \\$\ell_{\inf}=(\ell_{\inf,1},\dots,\ell_{\inf,N})$ and $\ell_{\sup}=(\ell_{\sup,1},\dots,\ell_{\sup,N})$.
\end{enumerate}

With either assumption, note that we cannot apply Theorem~\ref{th-1} directly since the linear equality constraint in~\eqref{eq:opt_dec_cent} is not full-row rank. However, if we instead consider partial contraction, then Theorem~\ref{th-1} can be used to prove the next result.
%

\begin{theorem}[Contraction analysis of distributed primal-dual dynamics]
\label{th-dec}
Consider the distributed primal-dual dynamics~\eqref{eq:pd_dec}. 
\begin{enumerate}
\item\label{desc-1-1} The distributed primal-dual dynamics are non-expansive with respect to $\norm{\cdot}_2$, and 
\item\label{desc-1-2}under Assumption~\ref{ass:5}, for any
  $(x^i(0),\nu^i(0))\in\R^n\times{\R^n}$,
  $\lim_{t\to\infty}x^i(t)=x^*$ and
  $\lim_{t\to\infty}\nu^i(t)=\nu^*_i$, for some $\nu_i^*$ such that $\sum^N_{k=1}\nu_k^*=\sum^N_{k=1}\nu^k(0)$.
\item \label{desc-2}Under Assumption~\ref{ass:6}, the convergence results in statement~\ref{desc-1-2} hold and, for $0<\epsilon<1$, the convergence of $(\textbf{x}(t),\nu(t))^\top$ has exponential rate 
\begin{equation}
\label{contr_dec_rate}
\frac{3\epsilon}{4}\frac{\lambda_N \lambda^2_2}{\lambda_N+1}\frac{\min_{i\in\until{N}}\ell_{\inf,i}}{\lambda_N^2+\frac{3}{4}\lambda_N \lambda^2_2+\norm{\ell_{\sup}}_{\infty}^2},
\end{equation}
where $\lambda_2$ and $\lambda_N$ are the smallest non-zero
and the largest eigenvalues of $L$, respectively.
%
%
%
%
\end{enumerate}
\end{theorem}
\begin{proof}
Set $f(\textbf{x})=\sum_{i=1}^Nf_i(x^i)$ and $\nu=(\nu^1,\dots,\nu^N)$. Succinctly, the dynamics of the system are
\begin{equation}
\label{suc-dec}
\begin{split}
\dot{\textbf{x}}&=-\nabla f(\textbf{x})-(L\otimes{I_n})\nu\\
\dot{\nu}&=(L\otimes{I_n})\textbf{x}
\end{split}.
\end{equation}
Now, let $\bar{A}:=(L\otimes I_n)$ and $(\dot{\textbf{x}},\dot{\nu}):=\fPDd(\textbf{x},\nu)$, 
and so $D\fPDd(\textbf{x},\nu)=\begin{bmatrix}
-\nabla^2f(\textbf{x}) & -\bar{A}^\top\\
\bar{A} & \vect{0}_{m\times m}
\end{bmatrix}$. 
Since $-\nabla^2f(\textbf{x})\preceq 0_{nN\times{nN}}$ because of convex $f_i$, it follows that $\mu_{2}(D\fPDd(\textbf{x},\nu))=0$ for any $\textbf{x}\in\R^{nN},\nu\in\R^{nN}$, and the system is weakly contractive, which proves~\ref{desc-1-1}.

Consider the equilibrium equations of~\eqref{suc-dec}, and let $(\textbf{x}^*,\nu^*)$ be a (candidate) fixed point of the system. From the second equation in~\eqref{suc-dec}, $\textbf{x}^*=\vect{1}_N\otimes v$ with $v\in\R^n$. Now, from the first equation in~\eqref{suc-dec}, we get 
$\vect{0}_{nN}=\nabla f(\textbf{x}^*)+(L\otimes{I_n})\nu^*$ and left multiplying by $\vect{1}_{N}^\top\otimes I_n$, we obtain $\vect{0}_{n}=\sum_{i=1}^N\nabla f_i(v)$. 
This is exactly the necessary and sufficient conditions of optimality for the problem~$\min_{x\in \R^{n}}f(x)=\sum_{i=1}^Nf_i(x)$, and so $v=x^*$ is an optimal solution to this problem. Moreover, $\nu^*$ is just some Lagrange multiplier for the constraint in~\eqref{eq:opt_dec_cent}.

Now, define the change of coordinates $(\textbf{x}',\nu')=(\textbf{x}-\textbf{x}^*,\nu-\nu^*)$, then we get 
$(\dot{\textbf{x}}',\dot{\nu}')=(\dot{\textbf{x}},\dot{\nu})=\fPDd(\textbf{x}'+\textbf{x}^*,\nu'+\nu^*)$, 
and whenever we refer to the word ``system" for the rest of the proof, we refer to the dynamics after this coordinate change. Observe the system has an equilibrium point $(\vect{0}_{nN},\vect{0}_{nN})$, but it is not unique; in fact, it is easy to verify that the following is a linear subspace of equilibria for the system: $\mathcal{M}=\setdef{(\textbf{x}',\nu')\in\R^{nN}\times\R^{nN}}{\textbf{x}'=\vect{0}_{nN},\nu'=\vect{1}_{N}\otimes\alpha \text{ with }\alpha\in\R^n}$. As a corollary, the subspace $\mathcal{M}$ is flow-invariant for the distributed system.

Now, since $L$ has $N-1$ strictly positive eigenvalues, we can write them as 
$0=\lambda_1<\lambda_2\leq\dots\leq \lambda_N$, and, by eigendecomposition, we can obtain an orthogonal matrix $R'\in\R^{N\times{N}}$ such that $R'LR'^\top=\diag(0,\lambda_2,\dots,\lambda_N)$. From here, we obtain the matrix $R\in\R^{N-1\times{N}}$ as a submatrix of $R'$ such that $RLR^\top=\Lambda$ with $\Lambda=\diag(\lambda_2,\dots,\lambda_N)$ with the properties: $RL=\Lambda R$, $RR^\top = I_{N-1}$ and $R^\top R\neq I_N$. Now, define $V=\diag(I_{nN},(R\otimes{I_n}))$ 
%
which has orthonormal rows and expresses $\mathcal{M}=\setdef{(\textbf{x}',\nu')\in\R^{nN}\times\R^{nN}}{V(\textbf{x}',\nu')^\top=\vect{0}_{(2N-1)n}}$. Then, 
we can use Theorem~\ref{th-contr-sub} for stating the convergence of trajectories of the system to $\mathcal{M}$ using partial contraction. 
First, note that 
$VD\fPDd(\textbf{x}',\nu') V^\top=
\begin{bmatrix}
\nabla^2f(\textbf{x}'+\textbf{x}^*)&-((R^\top\Lambda)\otimes{I_n})\\
((\Lambda R)\otimes{I_n})&\vect{0}_{n(N-1)\times{n(N-1)}}
\end{bmatrix}$,
where we have used the fact that $(R\otimes{I_n})(L\otimes
I_n)
=(\Lambda R)\otimes{I_n}$. Now, set
$\bar{A}^*:=(\Lambda R)\otimes{I_n}$ and note that
$\sigma_{\min}(\bar{A}^*) = \lambda_2$ and $\sigma_{\max}(\bar{A}^*) = \lambda_N$. 

Since $\nabla^2f(\textbf{x}'+\textbf{x}^*)\preceq \vect{0}_{nN\times{nN}}$, it follows that $\mu_{2}(VD\fPDd(\textbf{x}'+\textbf{x}^*,\nu'+\nu^*)V^\top)\leq 0$. Now, since $\nabla^2f(\textbf{x}^*)\prec \vect{0}_{nN\times{nN}}$, Proposition~\ref{prop-neg-eig} implies that $VD\fPDd(\textbf{x}^*,\nu^*)V^\top$ is a Hurwitz matrix, which implies that $\mu_{2}(VD\fPDd(\textbf{x}'+\textbf{x}^*,\nu'+\nu^*)V^\top)< 0$ for any $(\textbf{x}',\nu')\in\mathcal{M}$, and thus result~\ref{thss_1a} of Theorem~\ref{th-contr-sub} implies the convergence to the subspace $\mathcal{M}$ (and this implies convergence of the trajectories of the original system to the set $\mathcal{M}'=\setdef{(\textbf{x},\nu)\in\R^{nN}\times\R^{nN}}{\textbf{x}=\textbf{x}^*, \nu=\vect{1}_{N}\otimes\alpha \text{ with }\alpha\in\R^n}$, i.e., $\mathcal{M}$ is simply the set $\mathcal{M}'$ translated or anchored to the origin). Since $\mathcal{M}$ is a set of equilibria for the system and the system is weakly contractive, result~\ref{thss_2} of Theorem~\ref{th-contr-sub} concludes that any trajectory of the system converges to some equilibrium point in $\mathcal{M}$. 

Now, observe that $(\vect{1}_N^\top\otimes I_n)\dot{\nu}=(\vect{1}^\top_N L\otimes I_n)\textbf{x}=\vect{0}_{n}$, and so the set $\setdef{(\textbf{x},\nu)\in\R^{nN}\times\R^{nN}}{(\vect{1}_N^\top\otimes I_n)\nu=(\vect{1}_N^\top\otimes I_n)\nu(0)}$ is positively-invariant for~\eqref{eq:pd_dec}. Then, it follows that $\sum^N_{k=1}\nu^k_i(t)=\sum^N_{k=1}\nu^k_i(0)$ for any $i\in\until{n}$ and any $t\geq 0$. Then, since $\lim_{t\to\infty}\nu^i(t)<\infty$, 
we conclude the proof for statement~\ref{desc-1-2}.

We prove statement~\ref{desc-2}. Observe that 
\\$\min_{i\in\until{N}}\ell_{\inf,i}\;I_{nN\times{nN}}\preceq\nabla^2f(\textbf{x}'+\textbf{x}^*)\preceq \max_{i\in\until{N}}\ell_{\sup,i}\;I_{nN\times{nN}}$ and that $\bar{A}^*$ is full-row rank since it is easy to verify that $\rank(\Lambda R)=N-1$ and so $\rank((\Lambda R)\otimes{I_n})=n(N-1)$. Then, for $0<\epsilon<1$, defining $P= \begin{bmatrix} I_{nN}&\alpha_\epsilon\;
            \bar{A}^{*\top}\\ \alpha_\epsilon\; \bar{A}^* &I_{n(N-1)}
         \end{bmatrix} \in\real^{nN\times nN}$ and 
$\alpha_{\epsilon} :=  \frac{\epsilon\min_{i\in\until{N}}\ell_{\inf,i}}{\sigma^2_{\max}(\bar{A}^*)+\frac{3}{4}\sigma_{\max}(\bar{A}^*)\sigma^2_{\min}(\bar{A}^*)+\norm{\ell_{\sup}}_{\infty}^2}>0$, 
we can use Theorem~\ref{th-1} to conclude that 
$\mu_{2,P^{1/2}}(VD\fPDd(\textbf{x}'+\textbf{x}^*,\nu'+\nu^*)V^\top)\leq -c$ with $c$ as in equation~\eqref{contr_dec_rate} 
for any positive $\epsilon<1$.
%
%
So then, any trajectory $(\textbf{x}'(t),\nu'(t))$ exponentially
converges to the subspace $\mathcal{M}$ with rate $c$, due to
statement~\ref{thss_2} from Theorem~\ref{th-contr-sub}. 
Finally, the proof finishes by following a similar proof to statement~\ref{desc-1-2}.%
%
\end{proof}
For the case of convex $f_i$, Theorem~\ref{th-dec} does not state convergence - nor partial contraction - without additional assumptions. Similar to the analysis in Section~\ref{sec:basic-prob}, we present an example where augmenting the Lagrangian let us use Theorem~\ref{th-dec}.
%
%
We consider the popular 
\emph{distributed least-squares problem}~\cite{PW-SM-JL-WR:19}. Given a full-column rank matrix $H\in\R^{N\times{n}}$, $n<N$, it is known that $x^*=(H^\top H)^{-1}H^\top z$ is the unique solution to the least-squares 
problem $\min_{x\in\R^n}\norm{z-Hx}_2^2$, for $z\in\R^N$. 
An equivalent distributed version 
is
\begin{equation}\label{eq:ls-dec}
\begin{split}
	\min_{\textbf{x}\in \R^{nN}} \quad &\sum^N_{i=1}(h_i^\top x^i-z_i)^2\\
	&(L\otimes I_n)\textbf{x}= \vect{0}_{nN}\\
\end{split}
\end{equation}
with $h_i^\top\in\R^{1\times{n}}$ being the $i$th row of the matrix $H$, $\textbf{x}=(x^1,\dots,x^N)^\top$ and $z=(z_1,\dots,z_N)^\top$. Notice that $f(\textbf{x})=\sum^N_{i=1}|h_i^\top x^i-z_i|^2$ is convex, since 
$\nabla^2f(\textbf{x})=\diag(h_1h_1^\top,\dots, h_Nh_N^\top)\succeq\vect{0}_{nN\times{nN}}$. 
We propose to 
augment the Lagrangian 
with the quadratic term $\frac{\rho}{2}\textbf{x}^\top(L\otimes I_n)\textbf{x}$ with $\rho>0$ (which does not alter the original saddle points) 
and obtain
\begin{equation}\label{eq:pd_ls_1}
\begin{split}
\dot{x}^i&
	=-(h_i^\top x^i-z_i)h_i-\rho\sum_{j\in\mathcal{N}_i}(x^j-x^i)\\
	&-\sum_{j\in\mathcal{N}_i}(\nu^j-\nu^i)\\
\dot{\nu}^i&=\sum_{j\in\mathcal{N}_i}(x^j-x^i)
\end{split}
\end{equation}
for $i\in\until{N}$. The new algorithm is 
distributed.

Observe that $\ker(\diag(h_1h_1^\top,\dots, h_Nh_N^\top))\cap\ker(L\otimes I_n)=\{\vect{0}_{nN}\}$ implies \\$\ell_{\inf}^*\,I_{nN}\preceq \diag(h_1h_1^\top,\dots, h_Nh_N^\top)+(L\otimes I_n)$ for some constant $\ell_{\inf}^*>0$. Then, the following follows from Theorem~\ref{th-dec}. 

\begin{corollary}[Contraction analysis of distributed least-squares]
Consider the system~\eqref{eq:pd_ls_1}, and let $x^*$ be the unique solution to the least-squares problem. Then, 
for any $(x^i(0),\nu^i(0))\in\R^n\times{\R^n}$, $\lim_{t\to\infty}x^i(t)= x^*$ and $\lim_{t\to\infty}\nu^i(t)=\nu^*_i$ for some $\nu_i^*$ such that $\sum^N_{k=1}\nu_k^*=\sum^N_{k=1}\nu^k(0)$; and, for $0<\epsilon<1$, the convergence of $(\textbf{x}(t),\nu(t))$ has exponential rate
$$\epsilon\frac{3}{4}\frac{\lambda_N \lambda^2_2}{\lambda_N+1}\frac{\ell_{\inf}^*}{\lambda^2_N+\frac{3}{4}\lambda_N\lambda^2_2+\left(\lambda_N+\rho\max_i\norm{h_i}^2_2\right)^2}$$
where $\lambda_2$ and $\lambda_N$ are the smallest non-zero and
the largest eigenvalues of $L$, respectively.
\end{corollary}

\section{Time-varying optimization}
\label{sec:time-var-prob}

\subsection{Time-varying standard optimization}
\label{tv-p}
%
%
Our results in Section~\ref{sec:basic-prob} 
can be used to prove analyze 
%
the case where the associated optimization problem is time-varying. 
Consider  
\begin{equation}\label{eq:opt-onl}
	\min_{x\in \real^n} \; f(x,t)\quad\text{ subject to }\quad Ax= b(t)
%
\end{equation}
with the following standing assumptions: 
$A\in\R^{k\times n}$, $k<n$, $b\in\real^k$, $A$ is full-row rank,
and, for every $(x,t)\in\R^n\times{\R_{\geq 0}}$,
\begin{enumerate}
\item\label{as-1} $x\mapsto f(x,t)$ is twice continuously differentiable, uniformly strongly convex 
with constant $\ell_{\inf}>0$, i.e., $\nabla^2f(x,t)\succeq \ell_{\inf} I_n$; and uniformly Lipschitz smooth 
with constant $\ell_{\sup}>0$, i.e., $\nabla^2f(x,t)\preceq \ell_{\sup} I_n$;
\item \label{as-1a} $t\mapsto \nabla f(x,t)$ and $t\mapsto b(t)$ are continuously differentiable functions. 
\end{enumerate}
The associated \emph{time-varying primal-dual dynamics} are
\begin{equation}\label{eq:pd-online}
\begin{bmatrix}
\dot{x}\\\dot{\nu}
\end{bmatrix}=
\begin{bmatrix}
-\nabla f(x,t)-A^\top\nu\\
Ax-b(t)
\end{bmatrix}.
\end{equation}

Given a fixed time $t$, let $x^*(t)$ be a solution to the program $\min_{x:Ax=b(t)}f(x,t)$ and $\nu^*(t)$ its associated Lagrange multiplier. From the standing assumptions and Theorem~\ref{th-1}, for any fixed $t$, there exists a unique optimizer $(x^*(t),\nu^*(t))$. Then, $(x^*(t),\nu^*(t))_{t\geq 0}$ defines 
the \emph{optimizer trajectory} 
of the optimization problem~\eqref{eq:opt-onl}. The following result establishes the performance of the primal-dual dynamics in tracking the optimizer trajectory. 

\begin{theorem}[Contraction analysis of time-varying primal-dual dynamics]
\label{th:onl}
Consider the time-varying optimization problem~\eqref{eq:opt-onl}, its standing assumptions, and its associated primal-dual dynamics~\eqref{eq:pd-online}. 
\begin{enumerate}
\item\label{tv-1} The primal-dual dynamics are contractive with
  respect to $\norm{\cdot}_{2,P^{1/2}}$ with contraction rate $c$,
  where $P$ is the matrix defined in~\eqref{def:P} and $c$ is the same
  contraction rate as in~\eqref{contr_1} of
  Theorem~\ref{th-1}. 
\setcounter{saveenum}{\value{enumi}}
\end{enumerate}
Assume that, for any $t\geq 0$, $\norm{\dot{b}(t)}_2\leq \beta_1$ and
$\norm{\frac{\partial}{\partial t}\nabla f(x,t)}_2\leq \beta_2$ for some positive constants
$\beta_1,\beta_2$, and let $z(t):=(x(t),\nu(t))^{\top}$ and
$z^*(t):=(x^*(t),\nu^*(t))^{\top}$.
\begin{enumerate}\setcounter{enumi}{\value{saveenum}}
\item\label{tv-2}Then,
\begin{equation}
\begin{split}
\label{tracking-pd}
&\norm{z(t)-z^*(t)}_{2,P^{1/2}}\\
&\leq \left(\norm{z(0)-z^*(0)}_{2,P^{1/2}}-\frac{\rho}{c}\right)e^{-ct}+\frac{\rho}{c},
\end{split}
\end{equation}
i.e., the tracking error is uniformly ultimately bounded by $\frac{\rho}{c}$ with 
\begin{align*}
\rho=
\lambda_{\max}(P)&\Bigg(\frac{\beta_2}{\ell_{\inf}}+\left(\frac{\sigma_{\max}(A)}{\ell_{\inf}}+1\right)\\
&\frac{\ell_{\max}}{\sigma_{\min}^2(A)}\left(\beta_1+\frac{\sigma_{\max}(A)}{\ell_{\inf}}\beta_2\right)
\Bigg). 
\end{align*}
\end{enumerate}
\end{theorem}
\begin{proof}
Let $(\dot{x},\dot{\nu}):=\fPDo(x,\nu,t)$, and so 
$D\fPDo(x,\nu,t)=\begin{bmatrix}
-\nabla^2f(x,t) & -A^\top\\
A & \vect{0}_{k\times k}
\end{bmatrix}$. Since $A$ is constant and considering item~\ref{as-1} of the standing assumptions,   
we can finish the proof for statement~\ref{tv-1} by following the same proof as in Theorem~\ref{th-1}.
Now we prove statement~\ref{tv-2}. Let us fix any $t\geq 0$. Then, the KKT conditions that the optimizers $x^*(t)$ and $\nu^*(t)$ must satisfy (i.e., equivalent to the equilibrium equations of the system~\eqref{eq:pd-online}) are
\begin{align}
 \label{aca2}\vect{0}_n&=-\nabla f(x^*(t),t) - A^\top\nu^*(t)\\
\label{aca1}\vect{0}_k&=Ax^*(t)-b(t),
\end{align}
We first show that the curves $t\mapsto x^*(t)$ and $t\mapsto
\nu^*(t)$ are continuously differentiable. Define the function
$g:\real^{k+n+1}\to \real^{k+n}$ as 
$g(t,x,\nu) = \begin{bmatrix}
    -\nabla f(x,t) - A^{\top}\nu\\
    Ax - b(t)
  \end{bmatrix}$. 
Since $t\mapsto b(t)$ and $t\mapsto \nabla f(x,t)$ are continuously
differentiable, the function $g$ is continuously differentiable on
$\real^{n+k+1}$. Moreover, note that 
$\nabla_{(x,\nu)} g(t,x,\nu) = D\fPDo(x,\nu,t)$. 
By item~\ref{as-1} of the standing assumptions, we know that $-\nabla^2f(x,t)\preceq-\ell_{\inf}
I_n$ and $A$ is full row rank. From Proposition~\ref{prop-neg-eig}, this implies that $\nabla_{(x,\nu)} g(t,x,\nu)$ 
is Hurwitz and therefore, nonsingular. Finally, the Implicit Function
Theorem~\cite[Theorem 2.5.7]{RA-JEM-TSR:88} implies the solutions $t\mapsto x^*(t)$ and
$t\mapsto \nu^*(t)$ of the algebraic equations~\eqref{aca2} and~\eqref{aca1}
are continuously differentiable for any $t\in\R_{\geq 0}$. 

Now, observe that equation~\eqref{aca1} implies 
$\norm{A\dot{x^*}(t)}_2\leq \beta_1$.
Then, from~\eqref{aca2}
\begin{align*}
&\implies \vect{0}_{n}=\nabla^2f(x^*(t),t)\dot{x^*}(t)+\dot{\nabla}f(x^*(t),t)\\
&\qquad +A^\top\dot{\nu^*}(t)\\
&\implies \vect{0}_{m}=A\dot{x^*}(t)+A(\nabla^2f(x^*(t),t))^{-1}\dot{\nabla}f(x^*(t),t)\\
&\qquad+A(\nabla^2f(x^*(t),t))^{-1}A^\top\dot{\nu^*}(t)\\
&\implies \norm{\dot{\nu^*}(t)}_2\leq \frac{\ell_{\max}}{\sigma_{\min}^2(A)}\left(\beta_1+\frac{\sigma_{\max}(A)}{\ell_{\inf}}\beta_2\right),
\end{align*}
where the first implication follows from differentiation, the second one follows from the Hessian being  invertible, and the third one is derived considering that $A$ is full-row rank. Similarly, we  differentiate~\eqref{aca2} again and obtain 
\begin{align*}
\norm{\dot{x^*}(t)}_2\leq\frac{\beta_2}{\ell_{\inf}}+\frac{\sigma_{\max}(A)}{\ell_{\inf}}\norm{\dot{\nu^*}(t)}_2.
\end{align*}
\noindent Now, considering the contraction result on item~\ref{tv-1}, we set $\Delta(t):=\left\|\begin{bmatrix}x(t)\\ \nu(t)\end{bmatrix}-\begin{bmatrix}x^*(t)\\ \nu^*(t)\end{bmatrix}\right\|_{2,P^{1/2}}$ and use~\cite[Lemma~2]{HDN-TLV-KT-JJS:18} to obtain the following differential inequality $\dot{\Delta}(t)\leq -c\Delta(t)+\left\|\begin{bmatrix}\dot{x^*}(t)\\ \dot{\nu^*}(t)\end{bmatrix}\right\|_{2,P^{1/2}}$. Then, $\dot{\Delta}(t)\leq -c\Delta(t)+\lambda_{\max}(P)(\norm{\dot{x^*}(t)}_2+\norm{\dot{\nu^*}(t)}_2)$ 
and using our previous results, we immediately obtain $\dot{\Delta}(t)\leq -c\Delta(t)+\rho$ with 
$\rho$ as in the theorem statement.
%
%
%
%
%
Now, observe the function $h(u)=-cu+\rho$ is Lipschitz (since it is a linear function), and we can use the Comparison Lemma~\cite{HKK:02} to 
upper bound $\Delta(t)$ by the solution to the differential equation $\dot{u}(t)=-cu(t)+\rho$ for all $t\geq 0$, from which~\ref{tv-2} follows. 
%
%
%
%
%
\end{proof}
%
%
%
\begin{remark}
\label{remark_1}
The bounds in the assumptions for statement~\ref{tv-2} in Theorem~\ref{th:onl} ensure that 
the rate at which the time-varying
optimization changes is bounded. Indeed, the right-hand side of equation~\eqref{tracking-pd} is consistent: the larger (lower) these bounds, the larger (lower) the asymptotic tracking error. Moreover, the tracking is better the larger the contraction rate.
\end{remark}

\subsection{Time-varying distributed optimization}
\label{tv-do}

Our partial contraction analysis of Section~\ref{dec-dist} can be extended to 
obtain new results of performance guarantees 
for the following \emph{time-varying distributed optimization problem}
%
%
%
\begin{equation}\label{eq:opt-onl-d}
\begin{split}
	\min_{\textbf{x}\in \R^{nN}} \quad &\sum_{i=1}^Nf_i(x^i,t)\\
	&(L\otimes I_n)\textbf{x}= \vect{0}_{nN},
\end{split}
\end{equation}
where we consider a time-invariant connected undirected graph whose Laplacian matrix is $L$, and set $\textbf{x}=(x^i,\dots,x^N)^\top$ with $x^i\in\R^n$, 
with the following standing assumptions: 
for every $(x,t)\in\R^n\times{\R_{\geq 0}}$, and for any $i\in\until{N}$
\begin{enumerate}
\item\label{as-1-o-d} $x\mapsto f_i(x,t)$ is twice continuously differentiable, uniformly strongly convex with constant $\ell_{\inf,i}>0$, i.e., $\nabla^2f_i(x,t)\succeq \ell_{\inf,i} I_n$; and uniformly Lipschitz smooth with constant $\ell_{\sup,i}>0$, i.e., $\nabla^2f_i(x,t)\preceq \ell_{\sup,i} I_n$;
\item \label{as-1a-o-d} $t\mapsto \nabla f_i(x,t)$ is continuously differentiable.
\end{enumerate}
%
Then, the associated primal-dual dynamics are
\begin{equation}\label{eq:pd_dec_o}
\begin{split}
	\dot{x}^i&=-\nabla_{x^i}f_i(x^i,t)-\sum_{j\in\mathcal{N}_i}(\nu^j-\nu^i)\\
	\dot{\nu}^i&=\sum_{j\in\mathcal{N}_i}(x^j-x^i)\\
\end{split}
\end{equation}
%
%
for $i\in\until{N}$. Given a fixed time $t$, let $\textbf{x}^*(t)=\vect{1}_N\otimes x^*(t)$ with $x^*(t)$ being the unique solution to the program $\min_{x}\sum_{i=1}^Nf_i(x,t)$. Then, $(x^*(t))_t$ is a unique trajectory; however, there may exist multiple trajectories of the dual variables associated to the  
constraint in~\eqref{eq:opt-onl-d}. Let $\nu^*(t)=({\nu^1}^*(t),\dots,{\nu^N}^*(t))^\top$ be any dual variable obtained by solving the problem~\eqref{eq:opt-onl-d} for a fixed $t$. 
Then, we define the \emph{time-varying set of optimizers} as: 
%
\begin{align*}
&\mathcal{M}(t)=\{(\textbf{x},\nu)\in\R^{nN}\times\R^{nN}|\\
&V(\textbf{x}-\vect{1}_N\otimes{x^*(t)},\nu-\nu^*(t))^\top=(\vect{0}_{nN},\vect{0}_{n(N-1)})^\top\}
\end{align*}
where 
$V=\diag(I_{nN},R\otimes{I_n})$ with $R\in\R^{N-1\times N}$ as in the proof of Theorem~\ref{th-dec}. 
%
For convenience, let
$\ell_{\inf}=(\ell_{\inf,1},\dots,\ell_{\inf,N})$ and
$\ell_{\sup}=(\ell_{\sup,1},\dots,\ell_{\sup,N})$; and for $0<\epsilon<1$, we define 
\begin{equation}
 \begin{split}   
    &\tilde{\alpha}_{\epsilon} :=  \frac{\epsilon\min_{i\in\until{N}}\ell_{\inf,i}}{\lambda^2_N+\frac{3}{4}\lambda_N\lambda_2^2+\norm{\ell_{\sup}}_{\infty}^2}>0\\
    &\tilde{P} := \begin{bmatrix} I_{nN}&\alpha_\epsilon\;
            \bar{A}^{*\top}\\ \alpha_\epsilon\; \bar{A}^* &I_{n(N-1)}
         \end{bmatrix} \in\real^{nN\times nN}
         \label{def:P_onldpd}
  \end{split}
\end{equation}
where $\bar{A}^*=(\Lambda R)\otimes{I_n}$, with $\Lambda=\diag(\lambda_2,\dots,\lambda_N)$ containing the nonzero eigenvalues of $L$ in nondecreasing order. 
The following result establishes the performance of the primal-dual dynamics at tracking the time-varying set of optimizers.

\begin{theorem}[Contraction analysis of time-varying distributed primal-dual dynamics]
\label{th:onl-d}
Consider the time-varying optimization problem~\eqref{eq:opt-onl-d}, its standing assumptions, and its associated primal-dual dynamics~\eqref{eq:pd_dec_o}.
Set $z(t):=V(\textbf{x}(t),\nu(t))^\top$ and $z^*(t):=V(\textbf{x}^*(t),\nu^*(t))^\top$.
\begin{enumerate}
\item\label{tv-1-d}
The system associated to $\dot{z}$ is contractive with respect to $\norm{\cdot}_{2,\tilde{P}^{1/2}}$ with 
rate $c:=\alpha_\epsilon\frac{3}{4}\frac{\lambda^2_2}{\lambda_N+1}$. 
\setcounter{saveenum}{\value{enumi}}
\end{enumerate}
Moreover, for any $t\geq 0$, if $\norm{\frac{\partial}{\partial t}\nabla f_i(x,t)}_2\leq \beta_{1,i}$ for some positive constant $\beta_{1,i}$ and any $i\in\until{N}$, then, 
\begin{enumerate}\setcounter{enumi}{\value{saveenum}}
\item\label{tv-2-d}
\begin{equation}
\begin{split}
\label{tracking-pd-d}
&\norm{z(t)-z^*(t)}_{2,\tilde{P}^{1/2}}\\
&\leq\left(\norm{z(0)-z^*(0)}_{2,\tilde{P}^{1/2}}-\frac{\rho}{c}\right)e^{-ct}+\frac{\rho}{c},
\end{split}
\end{equation}
i.e., the tracking error is asymptotically bounded by $\frac{\rho}{c}$ with 
\begin{equation}
\begin{split}
\label{rho-dec-onl}
\rho&=\lambda_{\max}(P)\frac{\|\beta_{1}\|_{1}}{\|\ell_{\inf}\|_1}
N\\
&+\lambda_{\max}(P)\frac{\|\beta_{1}\|_1}{\lambda_2}\left(\frac{\|\ell_{\sup}\|_{\infty}}{\|\ell_{\inf}\|_1}+1\right).
%
\end{split}
\end{equation}
\end{enumerate} 
\end{theorem}
\begin{proof}
%
Define $f(\textbf{x}(t),t):=\sum^N_{i=1}f_i(x^i(t),t)$; 
then 
\begin{equation*} 
\dot{z}
=\begin{bmatrix}-\nabla f(\textbf{x}(t),t)-(L\otimes I_n)\nu(t)\\
(\Lambda R\otimes I_n)\textbf{x}(t)\end{bmatrix}.
\end{equation*}
Then, decomposing $(\textbf{x}(t),\nu(t))^\top=U(\textbf{x}(t),\nu(t))^\top+V^\top z$ where $U=I_{n(2N-1)}-V^\top V$ is a projection matrix, 
we use the chain rule and obtain that the Jacobian for this system  is
$$\begin{bmatrix}
-\nabla^2 f(\textbf{x}(t),t)&-(R^\top \Lambda\otimes I_n)\\ (\Lambda R\otimes I_n)&\vect{0}_{n(N-1)\times{n(N-1)}}
\end{bmatrix},$$
so then, based on our standing assumptions, using Proposition~\ref{prop-neg-eig} and following a similar proof to Theorem~\ref{th-dec}, we obtain that this system is contractive as in item~\ref{tv-1-d}.

Now we prove statement~\ref{tv-2-d}. The KKT conditions that the optimizers $\textbf{x}^*(t)$ and $\nu^*(t)$ must satisfy (i.e., the equilibrium equation of the system~\eqref{eq:pd_dec_o}), for any $t$, are
\begin{align}
\label{aca2-d}\vect{0}_{nN}&=-\nabla f(\textbf{x}^*(t),t) - (L\otimes I_n)\nu^*(t)\\
\label{aca1-d}\vect{0}_{nN}&=(L\otimes I_n)\textbf{x}^*(t).
\end{align}
Now, observe that \eqref{aca1-d} and~\eqref{aca2-d}$\implies \textbf{x}^*(t)=\vect{1}_N\otimes{x^*(t)}$ with $x^*(t)$ being the first $nN$ coordinates of any element of $\mathcal{M}(t)$. Moreover, by left multiplying~\eqref{aca1-d} with $\vect{1}_{N}^\top\otimes I_n$, we obtain that $\vect{0}_n=\sum_{i=1}^N\nabla_{x^i} f_i(x^*(t),t)$. Then, the Implicit Function Theorem~\cite[Theorem 2.5.7]{RA-JEM-TSR:88} (akin to its use in the proof of Theorem~\ref{th:onl}) implies the curve $t\mapsto x^*(t)$ is continuously differentiable for any $t\in\R_{\geq 0}$.
%

Now, from~\eqref{aca2-d} we obtain that $\vect{0}_{n(N-1)}=(R\otimes I_n)\nabla f(\textbf{x}^*(t),t) + (\Lambda\otimes I_n)(R\otimes I_n)\nu^*(t)$. Defining $y^*(t):=(R\otimes I_n)\nu^*(t)$, we get $\vect{0}_{n(N-1)}=(R\otimes I_n)\nabla f(\textbf{x}^*(t),t) + (\Lambda\otimes I_n)y^*(t)$. 
%
Again, an application of the Implicit Function Theorem let us conclude that 
the solution $(\textbf{x}^*,t)\mapsto y^*(\textbf{x}^*,t)$ is continuously differentiable for any $(\textbf{x}^*,t)\in\R^{nN}\times\R_{\geq 0}$; however, since 
$t\mapsto \textbf{x}^*(t)$ is continuously differentiable for any $t\in\R_{\geq 0}$, then 
$t\mapsto y^*(t)$ is continuously differentiable too. 

Then, we can differentiate equation~\eqref{aca2-d} and left multiply it by $(\vect{1}_N^\top\otimes{I_n})$ to obtain
%
\begin{align*}
\dot{x^*}(t)=-g(x^*(t),t)\sum_{i=1}^N\frac{\partial}{\partial t}\nabla_{x^i}f_i(x^*(t),t))
\end{align*}
with $g(x^*(t),t):=(\sum_{i=1}^N\nabla^2_{x_i}f_i(x^*(t),t))^{-1}$.
Recall that $RL=\Lambda R$. 
Then, since $y^*$ 
is continuously differentiable, we differentiate equation~\eqref{aca2-d} and left multiply it by $(R\otimes I_n)$ to obtain
%
\begin{align*}
\dot{y^*}(t)&=-(\Lambda^{-1}R\otimes I_n)(\nabla^2f(\textbf{x}^*(t),t)(\vect{1}_{N}\otimes h_1(x^*(t))\\
&+\frac{\partial}{\partial t}\nabla f(\textbf{x}^*(t),t)).
\end{align*}
Therefore, observe that 
$\norm{\dot{x^*}(t)}_2\leq\frac{1}{\sum_{i=1}^N\ell_{\inf,i}}\sum_{i=1}^N\beta_{1,i}=\frac{\|\beta_{1}\|_{1}}{\|\ell_{\inf}\|_1}$, and $\dot{\textbf{x}^*}(t)=\vect{1}_{N}\otimes \dot{x^*}(t)$. Moreover,
$\norm{\dot{y^*}(t)}_2\leq\frac{1}{|\lambda_2|}\left(\|\ell_{\sup}\|_{\infty}\norm{\dot{\textbf{x}^*}(t)}_2+\|\beta_{1}\|_{1}\right)$, 
where we used: 
$\norm{\frac{\partial}{\partial t}\nabla f(\textbf{x}^*(t),t)}_2\leq \sum_{i=1}^N\norm{\frac{\partial}{\partial t}\nabla_{x^i}f_i(x^*(t),t)}_2$, and $\norm{(\Lambda^{-1}R)\otimes I_n}_2=\sqrt{\lambda_{\max}(\Lambda^{-2}\otimes I_n)}
=\frac{1}{\lambda_{2}}$.

Now, for any $t$, let $(a_1(t),a_2(t))\in\mathcal{M}(t)$. Note that, no matter which element of $\mathcal{M}$ we choose, $a_1(t)=\vect{1}_N\otimes x^*(t)$ and so it is uniquely defined for any $t$ and we also know is differentiable. Now, note that $a_2(t)=\gamma(t)+\vect{1}_N\otimes\alpha$, with $\alpha\in\R^n$ and some uniquely defined $\gamma(t)$; and note that $(R\otimes I_n)a_2(t)=(R\otimes I_n)\gamma(t)$ for any $t$. Therefore $(R\otimes I_n)a_2(t)$ is uniquely defined for any $t$ and we also know is differentiable. In conclusion, the trajectory $((a_1(t),(R\otimes I_n)a_2(t)))_{t\geq 0}=\left(V
(a_1(t),a_2(t))^\top\right)_{t\geq 0}$ is unique and 
$t\mapsto V(a_1(t),a_2(t))^\top$ 
is a continuously differentiable curve. 

Since the system associated to 
$\dot{z}$ 
is contractive and the curve, as we just proved above, 
$t\mapsto z^*(t):=V(\textbf{x}^*(t),\nu^*(t))^\top$ 
is unique and differentiable, 
%
we set
$\Delta(t):=\|z(t)-z^*(t)\|_{2,P^{1/2}}$ 
    and use the result in item~\ref{tv-1-d} and~\cite[Lemma~2]{HDN-TLV-KT-JJS:18} to obtain the 
differential inequality 
%
%
\begin{align*}
\dot{\Delta}(t)&\leq -c\Delta(t)+\left\|\begin{bmatrix}\dot{\textbf{x}^*}(t)\\ \frac{d}{dt}\left((R\otimes I_n)\nu^*(t)\right)\end{bmatrix}\right\|_{2,P^{1/2}}\\
&\leq -c\Delta(t)+\lambda_{\max}(P)(N\norm{\dot{x^*}(t)}_2+\norm{\dot{z^*}(t)}_{2}).
\end{align*}
%
%
Finally, replacing our previous results and using the Comparison
Lemma~\cite{HKK:02} 
conclude the proof for~\ref{tv-2-d}.
\end{proof}

%
%
%
%
%
\begin{remark}
As in Remark~\ref{remark_1}, there is consistency on the right-hand side of equation~\eqref{rho-dec-onl}. 
\end{remark}

\section{Conclusion}
\label{sec:concl}

Primal-dual (PD) dynamics associated to linear equality constrained
optimization problems are studied in centralized, distributed and time-varying cases. Contraction theory provides an overarching analysis of the dynamical behavior and performance for all these cases of
PD 
dynamics.
%
%
As future work, we plan to design controllers that can improve the PD solver's tracking
properties in the time-varying 
setting. We also plan to study distributed PD solvers for globally coupled linear equation constraints and PD solvers in nonsmooth domains.


\bibliographystyle{plainurl}
\bibliography{alias,Main,FB}

\begin{thebibliography}{10}

\bibitem{RA-JEM-TSR:88}
R.~Abraham, J.~E. Marsden, and T.~S. Ratiu.
\newblock {\em Manifolds, Tensor Analysis, and Applications}, volume~75 of {\em
  Applied Mathematical Sciences}.
\newblock Springer, 2 edition, 1988.

\bibitem{ZA-EDS:14b}
Z.~Aminzare and E.~D. Sontag.
\newblock Contraction methods for nonlinear systems: {A} brief introduction and
  some open problems.
\newblock In {\em {IEEE} Conf.\ on Decision and Control}, pages 3835--3847,
  December 2014.

\bibitem{ZA-EDS:14}
Z.~Aminzare and E.~D. Sontag.
\newblock Synchronization of diffusively-connected nonlinear systems: {R}esults
  based on contractions with respect to general norms.
\newblock {\em IEEE Transactions on Network Science and Engineering},
  1(2):91--106, 2014.

\bibitem{KJA-LH-HU:58}
K.~J. Arrow, L.~Hurwicz, and H.~Uzawa, editors.
\newblock {\em Studies in Linear and Nonlinear Programming}.
\newblock Standford University Press, 1958.

\bibitem{MB-GHG-JL:05}
M.~Benzi, G.~H. Golub, and J.~Liesen.
\newblock Numerical solution of saddle point problems.
\newblock {\em Acta Numerica}, 14:1--137, 2005.

\bibitem{XC-NL:19}
X.~Chen and N.~Li.
\newblock Exponential stability of primal-dual gradient dynamics with
  non-strong convexity, 2019.
\newblock Arxiv preprint.
\newblock URL: \url{https://arxiv.org/pdf/1905.00298}.

\bibitem{AC-BG-JC:17}
A.~Cherukuri, B.~Gharesifard, and J.~Cortes.
\newblock Saddle-point dynamics: {Conditions} for asymptotic stability of
  saddle points.
\newblock {\em SIAM Journal on Control and Optimization}, 55(1):486--511, 2017.

\bibitem{PCV-SJ-FB:19r-arxiv}
P.~Cisneros-Velarde, S.~Jafarpour, and F.~Bullo.
\newblock Distributed and time-varying primal-dual dynamics via contraction
  analysis, 2020.
\newblock Arxiv preprint.
\newblock URL: \url{https://arxiv.org/pdf/2003.12665}.

\bibitem{WAC:1965}
W.~A. Coppel.
\newblock {\em Stability and Asymptotic Behavior Of Differential Equations}.
\newblock Heath, 1965.

\bibitem{JC-SKN:19}
J.~Cort{\'e}s and S.~K. Niederl{\"a}nder.
\newblock Distributed coordination for nonsmooth convex optimization via
  saddle-point dynamics.
\newblock {\em Journal of Nonlinear Science}, 29(4):1247--1272, 2019.

\bibitem{MdB-DF-GR-FS:16}
M.~{Di~Bernardo}, D.~Fiore, G.~Russo, and F.~Scafuti.
\newblock Convergence, consensus and synchronization of complex networks via
  contraction theory.
\newblock In J.~L{\"u}, X.~Yu, G.~Chen, and W.~Yu, editors, {\em Complex
  Systems and Networks: Dynamics, Controls and Applications}, pages 313--339.
  Springer, 2016.

\bibitem{MF-SP-VMP-AR:18}
M.~Fazlyab, S.~Paternain, V.~M. Preciado, and A.~Ribeiro.
\newblock Prediction-correction interior-point method for time-varying convex
  optimization.
\newblock {\em IEEE Transactions on Automatic Control}, 63(7):1973--1986, 2018.

\bibitem{DF-FP:10}
D.~Feijer and F.~Paganini.
\newblock Stability of primal--dual gradient dynamics and applications to
  network optimization.
\newblock {\em Automatica}, 46(12):1974--1981, 2010.

\bibitem{HKK:02}
H.~K. Khalil.
\newblock {\em Nonlinear Systems}.
\newblock Prentice Hall, 3 edition, 2002.

\bibitem{SK-JC-SM:15}
S.~S. Kia, J.~Cortes, and S.~Martinez.
\newblock Distributed convex optimization via continuous-time coordination
  algorithms with discrete-time communication.
\newblock {\em Automatica}, 55:254--264, 2015.

\bibitem{SL-LYW-GY:19}
S.~Liang, L.~Y. Wang, and G.~Yin.
\newblock Exponential convergence of distributed primal-dual convex
  optimization algorithm without strong convexity.
\newblock {\em Automatica}, 105:298--306, 2019.

\bibitem{YL-CL-BDOA-GS:19}
Y.~Liu, C.~Lageman, B.~D.O. Anderson, and G.~Shi.
\newblock An {Arrow}-{Hurwicz}-{Uzawa} type flow as least squares solver for
  network linear equations.
\newblock {\em Automatica}, 100:187--193, 2019.

\bibitem{YL-YL-DOA-GS:18}
Y.~Liu, Y.~Lou, B.~D.~O. Anderson, and G.~Shi.
\newblock Network flows that solve least squares for linear equations, 2018.
\newblock \href {http://arxiv.org/abs/1808.04140} {\path{arXiv:1808.04140}}.

\bibitem{WL-JJES:98}
W.~Lohmiller and J.-J.~E. Slotine.
\newblock On contraction analysis for non-linear systems.
\newblock {\em Automatica}, 34(6):683--696, 1998.

\bibitem{EL-GC-KS:14}
E.~{Lovisari}, G.~{Como}, and K.~{Savla}.
\newblock Stability of monotone dynamical flow networks.
\newblock In {\em {IEEE} Conf.\ on Decision and Control}, pages 2384--2389, Los
  Angeles, USA, December 2014.

\bibitem{EM-CZ-SL:14a}
E.~{Mallada}, C.~{Zhao}, and S.~{Low}.
\newblock Optimal load-side control for frequency regulation in smart grids.
\newblock {\em IEEE Transactions on Automatic Control}, 62(12):6294--6309,
  2017.

\bibitem{JM:00}
J.~Munkres.
\newblock {\em Topology}.
\newblock Pearson, 2 edition, 2000.

\bibitem{HDN-TLV-KT-JJS:18}
H.~D. {Nguyen}, T.~L. {Vu}, K.~{Turitsyn}, and J.~{Slotine}.
\newblock Contraction and robustness of continuous time primal-dual dynamics.
\newblock {\em IEEE Control Systems Letters}, 2(4):755--760, 2018.

\bibitem{QCP-JJS:07}
Q.~C. Pham and J.~J. Slotine.
\newblock Stable concurrent synchronization in dynamic system networks.
\newblock {\em Neural Networks}, 20(1):62--77, 2007.

\bibitem{GQ-NL:19}
G.~{Qu} and N.~{Li}.
\newblock On the exponential stability of primal-dual gradient dynamics.
\newblock {\em IEEE Control Systems Letters}, 3(1):43--48, 2019.

\bibitem{SR-WR:17}
S.~{Rahili} and W.~{Ren}.
\newblock Distributed continuous-time convex optimization with time-varying
  cost functions.
\newblock {\em IEEE Transactions on Automatic Control}, 62(4):1590--1605, 2017.

\bibitem{RTR:76}
R.~T. Rockafellar.
\newblock Augmented {Lagrangians} and applications of the proximal point
  algorithm in convex programming.
\newblock {\em Mathematics of Operations Research}, 1(2):97--116, 1976.

\bibitem{JWSP-BKP-NM-FD:19}
J.~W. {Simpson-Porco}, B.~K. {Poolla}, N.~{Monshizadeh}, and F.~{D\"orfler}.
\newblock Input-output performance of linear-quadratic saddle-point algorithms
  with application to distributed resource allocation problems.
\newblock {\em IEEE Transactions on Automatic Control}, 2019.

\bibitem{CS-MY-GH:17}
C.~Sun, M.~Ye, and G.~Hu.
\newblock Distributed time-varying quadratic optimization for multiple agents
  under undirected graphs.
\newblock {\em IEEE Transactions on Automatic Control}, 62(7):3687--3694, 2017.

\bibitem{MV:02}
M.~Vidyasagar.
\newblock {\em Nonlinear Systems Analysis}.
\newblock SIAM, 2002.

\bibitem{JW-NE:11}
J.~Wang and N.~Elia.
\newblock A control perspective for centralized and distributed convex
  optimization.
\newblock In {\em {IEEE} Conf.\ on Decision and Control and European Control
  Conference}, pages 3800--3805, Orlando, FL, USA, 2011.

\bibitem{PW-SM-JL-WR:19}
P.~Wang, S.~Mou, J.~Lian, and W.~Ren.
\newblock Solving a system of linear equations: {From} centralized to
  distributed algorithms.
\newblock {\em Annual Reviews in Control}, 47:306--322, 2019.

\bibitem{TY-XY-JW-DW:19}
T.~Yang, X.~Yi, J.~Wu, Y.~Yuan, D.~Wu, Z.~Meng, Y.~Hong, H.~Wang, Z.~Lin, and
  K.~H. Johansson.
\newblock A survey of distributed optimization.
\newblock {\em Annual Reviews in Control}, 47:278--305, 2019.

\end{thebibliography}

\section{Appendix}

\subsection{Proof of Proposition~\ref{prop-neg-eig}}

We remark that the proof of Proposition~\ref{prop-neg-eig} is complementary to the one given (for a slightly different case) in~\cite[Theorem~3.6]{MB-GHG-JL:05}).
\begin{proof}
Let $P:=\begin{bmatrix}
-B &-A^\top\\
A & \vect{0}_{m\times m}
\end{bmatrix}$ be the matrix in the proposition statement. First, note that $\Re(\lambda(P)) \le \mu_2(P)
= 0$. Therefore, every eigenvalue of $P$ has non-positive real part. 
We first show that $P$ has
no eigenvalue equal to zero. Note that by the Schur complement
determinant identity, we have $\det(P) = \det(-B))\det(-AB^{-1}A^{\top})$. 
Note that $B \succeq b_1 I_n$, therefore
$\det(-B) \ne 0$. Also, note that $B^{-1} \succeq b^{-1}_2 I_n$; and thus $A B^{-1}A^{\top} \succeq A(b^{-1}_{2}I_n) A^{\top} =b^{-1}_{2} AA^{\top} \succ 0$, 
where the last inequality follows from $AA^{\top}$ being invertible. 
This implies that $  \det(-A B^{-1}A^{\top}) \ne 0$. 
As a result, $\det(P)\ne 0$ and $P$ has no zero
eigenvalue. Now we show that $P$ is Hurwitz. Assume that $\lambda = \mathrm{i} \eta$
is an eigenvalue of $P$ with zero real part. This means
that, there exists $u\in \mathbb{C}^n$ and $v\in \mathbb{C}^m$ such that 
\begin{align}\label{eq:eigenvalue}
  \begin{bmatrix}-B & -A^{\top}\\
  A & \vect{0}_{m\times m}\end{bmatrix} \begin{bmatrix}u \\ v\end{bmatrix} = \mathrm{i} \eta \begin{bmatrix}u \\ v\end{bmatrix}. 
\end{align}
Multiplying this equation from the left by $[u^{H} , v^{H}]$,
we get 
$\Re\left(\begin{bmatrix}u^{H} &
  v^{H}\end{bmatrix}\begin{bmatrix}-B & -A^{\top}\\
  A & \vect{0}_{m\times m}\end{bmatrix} \begin{bmatrix}u \\v\end{bmatrix}\right) = 0$.  
This implies that $\Re\left(u^{H}Bu\right) = 0$. Assume
that $u =\theta_1 + \mathrm{i} \theta_2$, where $\theta_1,\theta_2\in \real^n$. Then
$\Re\left(u^{H}Bu\right) = 0$  is equivalent to $\theta_1^T B\theta_1 + \theta_2^{\top}B\theta_2 = 0$. 
Since $B \succeq b_1 I_n$, we get that $u =
\vect{0}_n$. As a result, the equation~\eqref{eq:eigenvalue} can be written as the system 
$A^{\top}v = \vect{0}_n$ and $v = \vect{0}_{k}$. 
This implies that if $\begin{bmatrix}u \\ v\end{bmatrix} \in
\mathbb{C}^{n+m}$ is an eigenvector associated to the eigenvalue
$\lambda = \mathrm{i}\eta$, then $\begin{bmatrix}u \\ v\end{bmatrix} =
\vect{0}_{n+m}$. Thus, the matrix $P$ has no eigenvalue
with zero real part. Therefore, the real part of every eigenvalue of
$P$ is negative and the matrix $P$ is Hurwitz.
\end{proof}

\subsection{A simple generalization of~\cite[Lemma~6]{EL-GC-KS:14}}

\begin{lemma}[Convergence of weakly contractive systems]
\label{prop-weak-conv}
Consider the dynamical system $\dot{x}=f(x,t)$, $x\in\R^n$, where $f$ is
continuously-differentiable with respect to $x$ and weakly contractive respect to some norm $\norm{\cdot}$, and let $x^*$ be an equilibrium
for the system, i.e., $f(x^*,t) =\vect{0}_n$, for every $t\ge 0$. Then $x^*$ is locally asymptotically stable if and only if it is globally asymptotically stable.
\end{lemma}
\begin{proof}
We only prove the nontrivial implication: if $x^*$ is locally asymptotically stable then it is globally asymptotically stable. Since $x^*$ is a locally asymptotically stable equilibrium point for
the dynamical system, then there exists $\epsilon>0$, such that, for every $y\in
\overline{B}(x^*,\epsilon)$ we have $\phi(t,0,y)\to x^*$ as $t\to \infty$. Note that,
for every $z\in \overline{B}(x^*,\epsilon)$, there exists $T_z$ such
that $\phi(T_z,0,z) \in
\overline{B}(x^*,\epsilon/2)$. Using the fact that the
closed ball $\overline{B}(x^*,\epsilon)$ is compact, we get that,
there exists $T$ such that, for every $z\in
\overline{B}(x^*,\epsilon)$, we have $\phi(T,0,z) \in
\overline{B}(x^*,\epsilon/2)$. Suppose that $t\mapsto x(t)$ is a
trajectory of the dynamical system. Assume
that $y\in \partial B(x^*,\epsilon)$ is a point on the straight line
connecting $x(0)$ to the unique equilibrium point
$x^*$. Then we have $\left\|x(T) - x^*\right\| \le  \left\|x(T) -
  \phi(T,0,y)\right\| + \left\|\phi(T,0,y)- x^*\right\|\le
  \left\|x(0)-y\right\| + \epsilon/2 =\|x(0)- x^*\|-\epsilon/2$.
Therefore, after time $T$, $t\mapsto \left\|x(t) -
  x^*\right\|$ decreases by $\epsilon/2$. As a result, there
    exists a finite time $T_{\inf}$ such that, for every $t\ge
    T_{\inf}$, we have $x(t)\in
    \overline{B}(x^*,\epsilon)$. Since
    $\overline{B}(x^*,\epsilon)$ is in the region of
    attraction of $x^*$ the trajectory $t\mapsto x(t)$
    converges to $x^*$.
%
\end{proof}

\end{document}